\documentclass{siamltex}

\usepackage{amsfonts, amsmath, amssymb, psfrag, graphicx, bbm, mathrsfs}

\usepackage[colorlinks=true]{hyperref}
\hypersetup{urlcolor=blue, citecolor=blue, linkcolor=blue}

\newtheorem{remark}[theorem]{Remark}

\def\fin{\ifmmode{\Large$\diamond$}\else{\unskip\nobreak\hfil
    \penalty50\hskip1em\null\nobreak\hfil{\Large$\diamond$}
    \parfillskip=0pt\finalhyphendemerits=0\endgraf}\fi}

\def\be#1#2\ee{\begin{equation}\label{eq:#1}#2\end{equation}}
\def\req#1{{\rm(\ref{eq:#1})}}
\def\bdm  {\begin{displaymath}}
  \def\edm  {\end{displaymath}}
\def\bdmal{\begin{displaymath}\begin{aligned}}
    \def\edmal{\end{aligned}\end{displaymath}}

\mathchardef\PhiG="0108

\renewcommand{\L}{{\mathscr L}}
\newcommand{\N}{{\mathord{\mathbb N}}}
\newcommand{\R}{{\mathord{\mathbb R}}}
\newcommand{\C}{{\mathord{\mathbb C}}}

\newcommand{\Z}{{\mathcal Z}}
\newcommand{\Ws}{W_\Sigma}

\newcommand{\bfh}{{\boldsymbol{h}}}

\newcommand{\B}{{\cal B}}

\newcommand{\Lrho}{{L_\varrho^\infty(\R^3)}}

\renewcommand{\L}{{\mathscr L}}

\newcommand{\norm}[1]{\|#1\|}

\newcommand{\Real}{{\rm Re\,}}

\newcommand{\rmd}{\,\mathrm{d}}

\newcommand{\rmi}{\mathrm{i}}
\newcommand{\eps}{\varepsilon}

\newcommand{\qbar}{{\overline q}}

\newcommand{\ahat}{{\widehat{a}}}
\newcommand{\omegahat}{{\widehat{\omega}}}

\newcommand{\Vhat}{{{\mathscr V}_{u^{\!\thsp\dagger}}}}

\newcommand{\zo}{{\overline{z}}}

\newcommand{\sigtilde}{{\widetilde y}}

\def\req#1{{\rm(\ref{eq:#1})}}

\makeatletter
\newcommand{\dupdots}{\mathinner{\mkern1mu\raise\p@
    \vbox{\kern7\p@\hbox{.}}\mkern2mu
    \raise4\p@\hbox{.}\mkern2mu\raise7\p@\hbox{.}\mkern1mu}}
\makeatother
\makeatletter
\newcount\c@MaxMatrixCols \c@MaxMatrixCols=10

\makeatother

\def\thsp{\hspace*{0.1ex}}
\def\thbsp{\hspace*{-0.2ex}}

\newcommand{\gC}{\mathcal{C}}
\newcommand{\gT}{\mathcal{T}}

\newcommand{\gCC}{\mathfrak{C}}
\newcommand{\gTT}{\mathfrak{T}}

 {\endMakeFramed}

\newcommand{\ud}{{u^\dagger\thbsp}}
\newcommand{\gd}{{g^\dagger\thbsp}}
\newcommand{\sd}{{y^\dagger\thbsp}}
\newcommand{\utilde}{{\widetilde u}}
\newcommand{\ftilde}{{\widetilde f}}

\newcommand{\omtilde}{{\widetilde\omega}}
\newcommand{\U}{{{\mathscr U}}}
\newcommand{\V}{{{\mathscr V}_u}}

\newcommand{\RR}{{\boldsymbol{R}}}

\newcommand{\dzeta}{\rmd \zeta}
\newcommand{\dtheta}{\rmd \theta}

\newcommand{\dR}{\rmd\thbsp R}
\newcommand{\dRR}{\rmd\!\RR}

\newcommand{\crho}{c_\varrho}
\newcommand{\cbeta}{c_\beta}

\newcommand{\Cbeta}{C_\beta}

\newcommand{\bfe}{{\boldsymbol{e_1}}}
\newcommand{\bfsigma}{{\boldsymbol{\sigma}}}

\makeatletter
\renewcommand\@biblabel[1]{#1.}
\makeatother

\title{Well-posedness of the Iterative Boltzmann Inversion}
\author{Martin Hanke\thanks{Institut f\"ur Mathematik, Johannes
    Gutenberg-Universit\"at Mainz, 55099 Mainz, Germany
    ({\tt hanke@math.uni-mainz.de}). The research leading to this work
    has been done within the 
    Collaborative Research Center TRR 146; corresponding funding 
    by the DFG is gratefully acknowledged.}}

\begin{document}
\sloppy
\maketitle

\begin{abstract}
The iterative Boltzmann inversion is an iterative scheme 
to determine an effective pair potential for an ensemble of identical particles
in thermal equilibrium from the corresponding radial distribution function. 
Although the method is reported to work reasonably well in practice, it still
lacks a rigorous convergence analysis.
In this paper we provide some first steps towards such an analysis, and we 
show under quite general assumptions that the algorithm is
well-defined in a neighborhood of the true pair potential, assuming that
such a potential exists. 

On our way we establish important properties
of the cavity distribution function and provide a proof of a statement 
formulated by Groeneveld concerning the rate of decay at infinity of the 
Ursell function associated with a Lennard-Jones type potential.
\end{abstract}


\begin{AMS}
  {\sc 82B21, 82B80}
\end{AMS}

\hspace*{-0.7em}
{\footnotesize \textbf{Last modified.} \today}

\pagestyle{myheadings}
\thispagestyle{plain}
\markboth{M. HANKE}
{ITERATIVE BOLTZMANN INVERSION}

\addtocounter{footnote}{1}

\section{Introduction}
\label{Sec:Introduction}
Numerical simulations of complex materials in physical chemistry
are so time-consuming, 
even with today's computing power at hand, that it is necessary to
implement subprocesses on a meso-scale by means of
``coarse-graining'' the atomistic structure of (parts of) the associated 
molecules.
The coarse-grained ``beads'' are simulated by using effective
potentials for their interactions. These
effective potentials have to be determined a priori, and this is often done 
so as to match some given structural data.

Here we consider the case where this structural information consists of
measurements of the so-called \emph{radial distribution function} $\gd$ 
of the beads; see~\req{g} for a formal definition of
this function. The effective potential $\ud$ is then chosen in such a way that
\be{Henderson-eq}
   F(\ud) \,=\, \gd\,,
\ee  
where $F$ is the function which maps a potential $u$
(out of a predetermined family of suitable functions) onto the 
corresponding radial distribution function of the associated
\emph{grand canonical ensemble} under well-defined physical conditions.
The question of existence and uniqueness of a solution
$\ud$ of \req{Henderson-eq} for a given $\gd$ is referred to as the
\emph{inverse Henderson problem},
because Henderson~\cite{Hend74} was the first to investigate the 
identifiability problem associated with \req{Henderson-eq}, i.e., 
whether the radial distribution function is enough data 
to uniquely recover the underlying pair potential; 
see Kuna, Lebowitz, and Speer~\cite{KLS07} for a more rigorous mathematical
treatment of the uniqueness problem.

A popular method for solving numerically the inverse Henderson problem is the
\emph{iterative Boltzmann inversion} (IBI) suggested by Soper~\cite{Sope96}. 
This method, which is available in pertinent public domain software like
{\sc votca}\footnote{{\tt http://www.votca.org}}~\cite{RJLKA09} 
starts from an initial guess $u_0$ 
and determines recursively an iterative sequence $(u_k)_{k\geq 0}$ 
of approximate solutions of \req{Henderson-eq} via
\be{IBI}
   u_{k+1} \,=\, u_k \,+\, \gamma \log \frac{F(u_k)}{\gd}\,, \qquad
   k=0,1,2,\dots\,.
\ee
Here $\gamma>0$ is a relaxation parameter that is usually chosen to be
\be{omega1}
   \gamma \,=\, 1/\beta\,,
\ee
where $\beta>0$ is the inverse temperature.

A mathematical analysis of the IBI method is still lacking although the method
seems to be fairly robust. In his original paper \cite{Sope96} Soper
provided a heuristic argument 
why IBI might be expected to converge, however, there is little hope to turn 
this argument into a rigorous proof.

From a mathematical point of view a possible framework for studying \req{IBI}
is fixed point iteration theory, where
\bdm
   \Phi(u) \,=\, u \,+\, \gamma \log \frac{F(u)}{\gd}
\edm
is the corresponding fixed point operator. This is the point of view
taken in this paper. Note, however, that currently we are not yet able to 
prove existence of a fixed point of $\Phi$ under reasonable assumptions on 
$\gd$ and the set of admissible pair potentials, although this
is evidently a necessary requirement for convergence of \req{IBI};
the only result in this direction that we are aware of has been 
provided by Koralov~\cite{Kora07}.

Instead we stipulate that a solution $\ud$ of \req{Henderson-eq} exists,
in which case we show that the first few iterates of~\req{IBI} are 
well-defined when the initial guess $u_0$ is chosen from a 
suitable neighborhood of $\ud$.
Additional assumptions that we impose are (i) that $\ud$ is a
\emph{Lennard-Jones type potential}, and
(ii) that the grand canonical ensemble is in the so-called \emph{gas phase};
see below.

The outline of this paper is as follows. In the following section we specify
our requirements on the family of admissible pair potentials. 
We then review in Section~\ref{Sec:Grandcanonical} the necessary background 
concerning the associated grand canonical ensemble and its thermodynamical 
limit, and investigate in more detail 
the so-called \emph{cavity distribution function}.
Section~\ref{Sec:Inequality} contains an auxiliary result on autoconvolution
products of a certain class of functions which include the 
\emph{Mayer function}. This will be applied in Section~\ref{Sec:Ursell} 
to discuss the rate of decay of the \emph{Ursell function} 
for large radii and to improve upon our earlier results in \cite{Hank16b} 
on the derivative of the Ursell function with respect to the pair potential. 
This provides the main ingredients of our analysis of the IBI method in 
Section~\ref{Sec:IBI}. 

We mention that although we treat IBI in the context of a 
grand canonical ensemble, it is possible to extend this analysis 
to a canonical ensemble, which is the more usual setting of numerical 
simulations in practice.

\section{Setting}
\label{Sec:Setting}
We start by considering an ensemble of identical classical particles in thermal
equilibrium within a bounded cubical box $\Lambda\subset\R^3$ centered at the
origin. We assume that the interaction of the particles
can be described by a pair potential $u=u(r)$, which 
only depends on the distance $r>0$ of the interacting particles,
and that this potential satisfies
\be{LJtype}
\begin{aligned}
   |u(r)| &\,\leq\, C r^{-\alpha}\,, \quad &r&\geq r_0\,, \\[1ex]
    u(r)\phantom{|} &\,\geq\, c r^{-\alpha}\,, \quad &r&\leq r_0\,,
\end{aligned}
\ee
for some $\alpha>3$, $r_0>0$, and parameters $c,C$ satisfying
$C_0>C>c>c_0>0$; here, $\alpha$, $r_0$, $c_0$, and $C_0$ are fixed constants, 
and we denote by $\U=\U(\alpha,r_0,c_0,C_0)$ 
the family of potentials $u$ satisfying~\req{LJtype}.
Following Ruelle~\cite{Ruel69} potentials $u\in\U$ are called
Lennard-Jones type pair potentials.

Under this assumption it is known 
(cf.~Fischer and Ruelle~\cite{FiRu66}) that there exists $B>0$ such that
\be{B}
   U_N(\RR_N) \,:=\!\! \sum_{1\leq i<j\leq N} \!\! u(|R_i-R_j|) \,\geq\, -BN
\ee
for every configuration of $N$ particles in free space and every $N\in\N$; 
here we denote by $R_i\in\R^3$ the coordinates of the $i$th particle, and
by $\RR_N=(R_1,\dots,R_N)\in(\R^3)^N$ the configuration of the first $N$
particles. It follows that for every $N\in\N$ and $\RR_N\in(\R^3)^N$ there is
an index $i^*(\RR_N)$ such that
\be{jstar}
   \sum_{i=1\atop \,\,\,i\neq i^*}^N u(|R_i-R_{i^*}|) \,\geq\, -2B\,.
\ee

Associated with $u\in\U$ and the inverse temperature $\beta>0$ is
the Mayer function
\be{Mayer}
   f(R) \,=\, e^{-\beta u(|R|)} - 1\,,
\ee
which is considered to be $-1$ at the origin $R=0$.
Because of \req{LJtype} the Mayer function is absolutely
integrable, i.e., there exists $\cbeta>0$ such that
\be{cbeta}
   \int_{\R^3} |f(R)|\dR \,<\, \cbeta\,.
\ee
By virtue of \req{B} and \req{cbeta}
every Lennard-Jones type potential is \emph{stable} and \emph{regular}
in the sense of \cite{Ruel69}.

As worked out in the proof of \cite[Proposition~2.1]{Hank16a}, 
for every $u\in\U$ the same constants $B$ and $\cbeta$ can be used in \req{B}
and \req{cbeta}, respectively, 
and also the indices $i^*$ in \req{jstar} can be chosen independent of $u\in\U$.

Associated with the constant $\alpha$ in \req{LJtype} is the weight function
\be{walpha}
   \varrho(r) \,=\, (1+r^2)^{\alpha/2}\,, \qquad r\geq 0\,.
\ee
For every $u\in\U$ we can use this weight function to define
a corresponding Banach space $\V$ of \emph{perturbations}, consisting of
all functions $v$ for which the associated norm
\be{Vd}
   \norm{v}_\V \,=\, 
   \max\bigl\{ \norm{v/u}_{(0,r_0]}, \norm{\varrho v}_{[r_0,\infty)} \bigr\}
\ee
is finite. Clearly, for every $u\in\U$ there exists 
$\delta_0=\delta_0(u)\in(0,1)$ 
sufficiently small such that $u+v\in\U$ for every $v\in\V$ with
$\norm{v}_\V\leq \delta_0$.
Later, compare \req{ustardef} and \req{ustar}, we will reduce the size of
$\delta_0(u)$ somewhat further to ensure additional properties of $u+v$
for all $v\in\V$ with $\norm{v}_\V\leq \delta_0$. 

\section{The grand canonical ensemble}
\label{Sec:Grandcanonical}
Let $u\in\U$ be the pair potential that determines the interaction
of the particles. 
In the grand canonical ensemble the number of particles and their coordinates 
within $\Lambda$ are random variables, and the probability of observing
an ensemble with exactly $N$ particles in an infinitesimal volume $\dRR_N$
at the coordinates $\RR_N\in\Lambda^N$
(up to permutations) is given by
\bdm
   \frac{1}{\Xi_\Lambda(z)}\,\frac{z^N}{N!}\,e^{-\beta U_N(\RR_N)}\dRR_N\,,
\edm
where $z>0$ is the so-called \emph{activity},
$U_N(\RR_N)$ is defined in \req{B}, and
\bdm
   \Xi_{\Lambda}(z) \,=\, \sum_{N=0}^\infty \frac{z^N}{N!}
         \int_{\Lambda^N} e^{-\beta U_N(\RR_N)}\dRR_N
\edm
is the grand canonical partition function. 
We consider this grand canonical ensemble under specified physical conditions,
i.e., we assume that the activity $z$ and the inverse temperature $\beta$
are given (and fixed), and that they satisfy the inequality
\be{Iz}
   0<z<\frac{1}{\cbeta e^{2\beta B+1}}\,,
\ee
where $\cbeta$ and $B$ are the constants in \req{cbeta} and \req{B}, 
respectively. 
This regime is known as the gas phase of the ensemble, cf.~\cite{Ruel69}.

For $m\in\N$ and $\RR_m\in\Lambda^m$ 
the $m$-particle distribution function is given by
\bdm
   \rho_{\Lambda}^{(m)}(\RR_m)
   \,=\, \frac{1}{\Xi_{\Lambda}(z)} \sum_{N=m}^\infty \frac{z^N}{(N-m)!}
         \int_{\Lambda^{N-m}} e^{-\beta U_N(\RR_N)}\dRR_{m,N}\,,
\edm
where $\RR_{m,N}=(R_{m+1},\dots,R_N)$; $\rho_\Lambda^{(m)}$ determines the 
probability distribution for snap shots with $m$ particles 
(up to permutations) at coordinates $R_1,\dots,R_m\in\Lambda$.\footnote{If 
$R_i=R_j$ for different indices $i,j\in\{1,\dots,m\}$ then
$\rho_\Lambda^{(m)}(\RR_m)$ is set to be zero.} 
The grand canonical partition function $\Xi_\Lambda$ can be seen to be
an entire function of $z\in\C$, which is free of zeros for
\bdm
   z\in\Z\,=\,
   \bigl\{\,z\in\C\,:\, |z| \,<\, \frac{1}{\cbeta e^{2\beta B+1}}\,
   \bigr\}\,,
\edm
compare \cite[Theorem~4.2.3]{Ruel69}, and similarly, the
$m$-particle distribution functions all are analytic functions of $z\in\Z$.
On the other hand, the $m$-particle distribution functions may encounter
singularities for positive values of $z$ outside the interval~\req{Iz}; 
those are understood to correspond to physical phase transitions.

As shown in \cite{Ruel69} the $m$-particle distribution function 
has a well-defined \emph{thermodynamical limit}, i.e., 
$\rho_{\Lambda}^{(m)}$ converges to some $\rho^{(m)}\in L^\infty\bigl((\R^3)^m\bigr)$
as $|\Lambda|\to\infty$, uniformly on every compact subset of $(\R^3)^m$
and for activities $z$ from every compact subset of $\Z$, 
this being true for every $m\in\N$; 
here, $|\Lambda|$ denotes the volume of the box.
In particular, for $m=1$, the thermodynamical limit
\bdm
   \rho^{(1)}(R) \,=\, \rho_0 \,\in\R_0^+
\edm
is independent of $R\in\R^3$ and provides the counting density of the ensemble;
for $m=2$, $\rho^{(2)}(R_1,R_2)$ only depends on the distance 
$r=|R_1-R_2|\geq 0$.
Given these two functions the radial distribution function, 
referred to in the introduction, is defined to be
\be{g}
   g(r) \,=\, \frac{1}{\rho_0^2}\,\rho^{(2)}(R,0)\,, \qquad |R|=r\geq 0\,.
\ee

For $m\in\N_0$ and $z\in\Z$ the function $\rho^{(m)}$ is 
Fr\'echet differentiable
with respect to $u$, i.e., with respect to perturbations $v\in\V$ of $u$, 
cf.~\cite{Hank16a}. The derivative is a bounded linear operator
$\partial\rho^{(m)}\in\L\bigl(\V,L^\infty((\R^3)^m)\bigr)$.
A similar result (see \cite[Remark~3.4]{Hank16a} for details)
applies to certain weighted copies of the particle distribution functions, 
which include the cavity distribution function 
(cf.~Hansen and McDonald~\cite{HaMcD13})
\be{y}
   y(r) \,=\, e^{\beta u(r)} g(r)\,, \qquad r>0\,,
\ee
as a special case; the following result elaborates on this. 

\begin{proposition}
\label{Prop:rho2-low}
For $u\in\U$ and $z\in\Z$ the cavity distribution function $y$ of \req{y} is
a bounded function of $r>0$, which is analytic with respect to $z\in\Z$ 
and uniformly bounded on every compact subset of $\Z$.
Moreover, $y$ is Fr\'echet differentiable with respect to
$u$ with derivative $\partial y \in \L(\V,L^\infty(\R^+))$.
For
\be{z-bound}
   0 < z \leq \zo < \frac{1}{1+e}\frac{1}{\cbeta e^{2\beta B+1}}
\ee
the cavity distribution function is strictly positive, i.e.,
there exists $c>0$ (depending only on $\zo$) such that 
\be{sigmastar}
   y(r) \,\geq\, c\,\frac{z^2}{\rho_0^2}\,, 
   \qquad r>0\,,
\ee
for all $u\in\U$.
\end{proposition}

\begin{proof}
The function
\bdm
   \sigma^{(2)}(R,0) \,=\, \rho_0^2\,y(|R|)
\edm
is the second entry of the semi-infinite vector $\bfsigma=\bfsigma_{\R^3}$ 
considered in \cite[Remark~3.4]{Hank16a}. There it is shown that $\bfsigma$
satisfies a system 
\be{KirkwoodSalsburg}
   (I-zB)\bfsigma \,=\, z\bfe\,, \qquad 
   B=KD\,,
\ee
of Kirkwood-Salsburg integral equations, and that $\bfsigma$ has certain
differentiability properties. These properties readily imply differentiability
of $y$ with respect to $z$ and $u$ as stated above.

In \req{KirkwoodSalsburg} we have adopted notation of \cite{Hank16a}:
$K$ is a semi-infinite matrix of integral operators, $D$ a diagonal 
multiplication operator, and $I$ the corresponding identity operator;
$\bfe$ is a vector of constant functions, its first entry being identically 
one, and all other entries being zero. 
To establish the lower bound~\req{sigmastar} for
the specific real interval of activity parameters $z$ given in \req{z-bound},
we first note that $I-zB$ can be developed into a Neumann series,
and hence we can rewrite \req{KirkwoodSalsburg} in the form
\be{3.sigma}
   \bfsigma 
   \,=\, z\bfe \,+\, z^2B\bfe \,+\, \bfh\,,
\ee
where
\bdm
   \bfh \,=\, (h^{(m)})_m
   \,=\, z^3(I-zB)^{-1}B KD\bfe
   \,=\, z^3(I-zB)^{-1}B K \bfe\,,
\edm
because the $(1,1)$-entry of $D$ is an identity operator.
Looking at the second entry of the vector identity~\req{3.sigma} 
we conclude that
\be{KStmp}
   \sigma^{(2)}(R,0) - z^2 \,=\, h^{(2)}(R,0)\,,
\ee
because the second entry $b_{21}$ of $B\bfe$ is again a constant, i.e.,
$b_{21}=1$; compare~\cite{Hank16a}.
For $z\in\Z$ the right-hand side of \req{KStmp} 
can be bounded as in \cite{Hank16a}, which gives
\bdm
   \norm{h^{(2)}}_{L^\infty((\R^3)^2)} \,\leq\, 
   z^3 e \,\frac{\cbeta e^{2\beta B+1}}{1-z\cbeta e^{2\beta B+1}}\,.
\edm
Therefore there exists $c>0$ such that
for every $R\in\R^3$ with $|R|=r>0$ there holds
\bdm
   y(r) \,=\, \frac{\sigma^{(2)}(R,0)}{\rho_0^2}
   \,\geq\, \frac{1}{\rho_0^2}\,z^2 
            \Bigl(1-ze\,\frac{\cbeta e^{2\beta B+1}}{1-z\cbeta e^{2\beta B+1}}\Bigr)
   \,\geq\, c\,\frac{z^2}{\rho_0^2} \,,
\edm
provided $0<z\leq\zo$. 
\end{proof}

\section{An auxiliary inequality for autoconvolution products}
\label{Sec:Inequality}
Before we continue we define the Banach space $L_\varrho^\infty(\R^3)$ 
of functions $v\in L^\infty(\R^3)$ with finite norm
\be{Linftyrho}
   \norm{w}_{L_\varrho^\infty(\R^3)} \,=\, \sup_{R\in\R^3} \varrho(|R|)|w(R)|\,,
\ee
where $\varrho$ is as in \req{walpha}. We mention that the Mayer $f$-function
defined in \req{Mayer} belongs to this space by virtue of \req{LJtype}.
Note that $L_\varrho^\infty(\R^3)$ is continuously embedded in $L^1(\R^3)$
and $L^\infty(\R^3)$, because the parameter $\alpha$ in \req{LJtype} is assumed 
to satisfy $\alpha>3$. It readily follows that the convolution $w* w'$ 
of two functions $w,w'\in\Lrho$ is an absolutely integrable function. 
In fact, we show next that the result belongs to $\Lrho$ again.

\begin{proposition}
\label{Prop:Banachalgebra}
Let $w,w'\in\Lrho$. Then $w* w'\in\Lrho$ with
\be{Banachalgebra}
   \norm{w* w'}_\Lrho \,\leq\, \crho 2^{\alpha+1} \norm{w}_\Lrho \norm{w'}_\Lrho\,,
\ee
where $\crho$ is the embedding constant for the embedding of
$\Lrho$ into $L^1(\R^3)$. 
\end{proposition}

\begin{proof}
For $R\in\R^3$ and $0<\eps<1$ we consider the ball $\B_{\eps|R|}(R)\subset\R^3$ 
of radius $\eps|R|$ around $R$. Depending on whether $R'$ is inside or outside
this ball, there holds
\begin{subequations}
\label{eq:eps-alpha-tmpX}
\begin{align}
\label{eq:eps-alpha-tmpX1}
    1 + |R'|^2 & \,\geq\, (1-\eps)^2 (1+|R|^2)\,,
    & R'&\in\B_{\eps|R|}(R)\,,\\[1ex]
\label{eq:eps-alpha-tmpX2}
   1+|R'-R|^2 & \,\geq\, \eps^2(1+|R|^2)\,, & 
   R'&\in\R^3\setminus\B_{\eps|R|}(R)\,.
\end{align}
\end{subequations}
Using \req{eps-alpha-tmpX1} it follows for every $R\in\R^3$ that
\begin{align*}
   &\varrho(|R|)\, 
   \Biggl| \int_{\B_{\eps|R|}(R)} w(R-R') w'(R')\dR'\,\Biggr| \\[1ex]
   &\qquad
    \,\leq\, \int_{\B_{\eps|R|}(R)} \frac{\varrho(|R|)}{\varrho(|R'|)} \,
             \bigl|w(R-R')\bigr|\, \varrho(|R'|)\,\bigl|w'(R')\bigr|\dR'\\[1ex]
   &\qquad
    \,\leq\, \frac{1}{(1-\eps)^\alpha}\, \norm{w'}_{L_\varrho^\infty(\R^3)}
             \int_{\B_{\eps|R|}(R)} \bigl|w(R-R')\bigr|\dR'\\[1ex]
   &\qquad
    \,\leq\, \frac{1}{(1-\eps)^\alpha}\, \norm{w}_{L^1(\R^3)}
             \norm{w'}_{L_\varrho^\infty(\R^3)}\,,
\intertext{while \req{eps-alpha-tmpX2} implies that}
   &\varrho(|R|) \,
    \Biggl| \int_{\R^3\setminus\B_{\eps|R|}(R)} w(R-R') w'(R')\dR'\,\Biggr| \\[1ex]
   &\qquad
    \,\leq\, \int_{\R^3\setminus\B_{\eps|R|}(R)} 
             \frac{\varrho(|R|)}{\varrho(|R-R'|)}\,
             \varrho(|R-R'|)\,\bigl|w(R-R')\bigr|\,\bigl|w'(R')\bigr|\dR'\\[1ex]
   &\qquad
    \,\leq\, \frac{1}{\eps^\alpha}\,\norm{w}_{L_\varrho^\infty(\R^3)}
             \int_{\R^3\setminus\B_{\eps|R|}(R)} \bigl|w'(R')\bigr|\dR' 
    \,\leq\, \frac{1}{\eps^\alpha}\,
             \norm{w}_{L_\varrho^\infty(\R^3)}\,\norm{w'}_{L^1(\R^3)}\,.
\end{align*}
Adding these two inequalities we thus conclude that
\begin{align}
\label{eq:Wp-recX}
   \norm{w* w'}_{L_\varrho^\infty(\R^3)}
   &\,\leq\, \frac{1}{(1-\eps)^\alpha}\,\norm{w}_{L^1(\R^3)}
             \norm{w'}_{L_\varrho^\infty(\R^3)}
             \,+\, \frac{1}{\eps^\alpha}\,
                   \norm{w}_{L_\varrho^\infty(\R^3)}\,\norm{w'}_{L^1(\R^3)}\\[1ex]
\nonumber
   &\,\leq\, \Bigl(\frac{1}{(1-\eps)^\alpha} + \frac{1}{\eps^\alpha}\Bigr)
             \crho \norm{w}_\Lrho \norm{w'}_\Lrho\,,
\end{align}
where $\crho$ is the embedding constant for the embedding 
$\Lrho\subset L^1(\R^3)$. 
By choosing $\eps=1/2$ we finally obtain \req{Banachalgebra}.
\end{proof}

Proposition~\ref{Prop:Banachalgebra} implies that we can rescale the norm
of $\Lrho$ to make $\Lrho$ a commutative Banach algebra with respect to 
convolution. 

For $w\in\Lrho$ and $n\in\N$ let $W_n$ be the $n$-fold autoconvolution of $w$, 
i.e.,
\be{Uell}
   W_1=w\,, \qquad W_{n+1} = w*W_n\,, \quad n\geq 1\,.
\ee
By virtue of Proposition~\ref{Prop:Banachalgebra} each $W_n$ belongs to 
$\Lrho$, and there holds
\be{Wp-L1}
   \norm{W_n}_{L^1(\R^3)} \,\leq\, \norm{w}_{L^1(\R^3)}^n\,, \qquad n\geq 1\,.
\ee

\begin{proposition}
\label{Prop:Wp-bound}
Assume that $w\in\Lrho$ satisfies
\be{L1normf}
   \norm{w}_{L^1(\R^3)} \,=\, q \,<\,1\,,
\ee
and let $\qbar\in(q,1)$. Then the autoconvolution products $W_n$
defined in \req{Uell} satisfy
\be{Wp-bound}
   \norm{W_n}_{L_\varrho^\infty(\R^3)} \,\leq\, C_*\,\qbar^n \norm{w}_\Lrho\,, \qquad
   n\in\N\,,
\ee
for some constant $C_*>0$ depending only on $\alpha$, $q$, and $\qbar$.
\end{proposition}

\begin{proof}
We are going to prove by induction the inequality
\be{Wp-bound-prime}
   \norm{W_n}_{L_\varrho^\infty(\R^3)}
   \,\leq\, \frac{1}{\eps^\alpha}\,\norm{w}_{L_\varrho^\infty(\R^3)}\,
            \frac{1-(q/\qbar)^n}{1-q/\qbar}\,\qbar^{n-1}\,,
\ee
where we let
\be{eps-alpha}
   \eps \,=\, 1-(q/\qbar)^{1/\alpha}\,,
\ee
which is a positive number; this
readily implies \req{Wp-bound}.
The induction base $n=1$ of \req{Wp-bound-prime} is obviously correct
because $\eps<1$ according to \req{eps-alpha}. For the induction step 
from $n$ to $n+1$, $n\geq 1$, we apply inequality~\req{Wp-recX} from the 
proof of Proposition~\ref{Prop:Banachalgebra} with $w'=W_n$ and $\eps$
of \req{eps-alpha} to obtain
\bdm
   \norm{W_{n+1}}_{L_\varrho^\infty(\R^3)}
   \,\leq\, \qbar\,\norm{W_n}_{L_\varrho^\infty(\R^3)}
            \,+\, \frac{1}{\eps^\alpha}\,
                  \norm{w}_{L_\varrho^\infty(\R^3)}\,\norm{W_n}_{L^1(\R^3)}\,.
\edm
Inserting \req{Wp-L1}, \req{L1normf}, and the induction hypothesis 
\req{Wp-bound-prime} this yields
\bdm
   \norm{W_{n+1}}_{L_\varrho^\infty(\R^3)}
   \,\leq\, \frac{1}{\eps^\alpha}\,\norm{w}_{L_\varrho^\infty(\R^3)}
            \Bigl(
               \frac{1-(q/\qbar)^n}{1-q/\qbar} \,+\, (q/\qbar)^n
            \Bigr)\qbar^n\,,
\edm
which coincides with the bound \req{Wp-bound-prime} for the norm of $W_{n+1}$.
\end{proof}

\begin{corollary}
\label{Cor:Wp-bound}
Under the assumptions of Proposition~\ref{Prop:Wp-bound} the infinite
series
\be{Qell}
   \Ws \,=\, \sum_{n=1}^\infty W_n
\ee
converges in $L_\varrho^\infty(\R^3)$.
\end{corollary}



\section{The Ursell function}
\label{Sec:Ursell}
The Ursell function (relative to the origin) of our grand canonical ensemble
with pair potential $u\in\U$
(see Section~\ref{Sec:Grandcanonical}) is defined to be
\be{Ursell}
   \omega_\Lambda(R)
   \,=\, \rho_\Lambda^{(2)}(R,0)-\rho_\Lambda^{(1)}(R)\rho_\Lambda^{(1)}(0)\,,
   \qquad R\in\Lambda\,.
\ee
For $z\in\Z$ the Ursell function can be expanded
into an absolutely convergent power series 
\be{an}
   \omega_\Lambda(R) \,=\, \sum_{N=2}^\infty a_{N,\Lambda}(R)\thsp z^N,
\ee
the coefficients of which depend on $u$ via the Mayer function $f$ defined
in \req{Mayer}. They can be represented in the form
\begin{subequations}
\label{eq:an-graph}
\be{an-graph-N}
   a_{N,\Lambda}(R) \,=\, \frac{1}{(N-2)!} 
   \int_{\Lambda^{N-2}}
      \sum_{\gC\in\gCC_N}\prod_{(i,j)\in\gC}\!\!\!f(R_i-R_j)\thsp
   \dRR_{2,N}\,,
\ee
cf.~Stell~\cite{Stel64},
where $R_1=R$ and $R_2=0$, 
$\gCC_N$ is the set of connected graphs with $N$ vertices,
labeled $1,\dots,N$,
and the product in \req{an-graph-N} runs over all bonds in $\gC$: 
the notation $(i,j)$ refers to a bond connecting vertices 
$i$ and $j$, where we use the convention that $i<j$ for $(i,j)\in\gC$.
For $N=2$ the representation~\req{an-graph-N} is to be read as
\be{an-graph-2}
   a_{2,\Lambda}(R) \,=\, f(R)\,.
\ee
\end{subequations}

For our reference potential $u$ we consider a perturbation $v\in\V$
with $\norm{v}_\V\leq\delta_0(u)$, compare Section~\ref{Sec:Setting}.
For the \emph{complex} pair potential $\utilde=u+\zeta v$ with 
$\zeta\in\C$, $|\zeta|\leq 1$, we define $\widetilde{U}_N$, $V_N$, and $|V|_N$
as in \req{B}, replacing $u$ by $\utilde$, $v$ and $|v|$, respectively,
on the right-hand side.
Using this notation it follows from \req{B} that $\utilde$ satisfies the 
stability bound
\be{Ueltschi-stab-bound}
\begin{aligned}
   \Real\bigl(\widetilde{U}_N(\RR_N)\bigr)
   &\,=\, U_N(\RR_N) \,+\, \Real(\zeta)\, V_N(\RR_N)\\[1ex]
   &\,\geq\, U_N(\RR_N) \,-\, |V|_N(\RR_N) \,\geq\, -NB
\end{aligned}
\ee
for every coordinate vector $\RR_N\in(\R^3)^N$, because $u-|v|\in\U$ by the
definition of $\delta_0$. Furthermore, the associated (complex) Mayer function
\bdm
   \ftilde(R) \,=\, e^{-\beta\utilde(|R|)}-1\,, \qquad R\in\R^3\,, 
\edm
satisfies 
\bdm
   \bigl|\ftilde(R) - f(R)\bigr| 
   \,\leq\, \beta e^{-\beta (u-|v|)(|R|)} \bigl|v(|R|)\bigr| 
   \,\leq\, \left\{ \!\!\!
            \begin{array}{rl}
               \dfrac{1}{e(1-\delta_0)} \norm{v}_\V\,, & 
               |R| < r_0\,,\\[2.5ex]
               \dfrac{\beta e^{2\beta B}}{\varrho(|R|)}\, \norm{v}_\V\,, & 
               |R|\geq r_0\,.
            \end{array}
            \right.
\edm
It follows that
\be{ftildew}
   |\ftilde(R)| \,\leq\, \cbeta w(R), \qquad R\in\R^3\,,
\ee
where $\cbeta$ has been introduced in \req{cbeta}, and
\be{ustardef}
   w(R) \,=\, \frac{1}{\cbeta}\,|f(R)| 
              \,+\, \Cbeta\,\frac{\delta_0}{\varrho(|R|)}
\ee
for some suitably chosen constant $\Cbeta>0$; 
note that $w$ is a positive function. 
Reducing the size of $\delta_0$, when necessary, we can make sure that
\be{ustar}
   q \,:=\, \int_{\R^3} w(R)\dR \,<\, 1
\ee
by virtue of \req{cbeta}.
We fix $\delta_0=\delta_0(u)$ accordingly for the remainder of this paper.
Note that \req{ftildew} holds uniformly 
for all $\utilde=u+\zeta v$ with $\norm{v}_\V\leq\delta_0$
and $\zeta\in\C$, $|\zeta|\leq 1$.

As we have already mentioned in Section~\ref{Sec:Inequality} the 
Mayer $f$-function belongs to the Banach space $\Lrho$ introduced 
in \req{Linftyrho},
hence the function $w$ of \req{ustardef} satisfies the
assumptions of Proposition~\ref{Prop:Wp-bound} and 
Corollary~\ref{Cor:Wp-bound}: As before we denote by $W_n,\Ws\in\Lrho$ the 
corresponding autoconvolutions~\req{Uell} and their infinite series~\req{Qell},
respectively. Note that
\be{monotonicity}
   0 \,<\, W_n(R) \,\leq\, \Ws(R)
\ee
for every $R\in\R^3$ and $n\in\N$, because $w$ is a positive function.

Now we fix a perturbation $v_0\in\V$ with $\norm{v_0}_\V\leq\delta_0$, 
an activity parameter $z$ satisfying \req{Iz},
and coordinates $R_1=R\in\Lambda$ and $R_2=0$. 
Let $N\geq 2$ and 
$\RR_{2,N}\in\Lambda^{N-2}$ be further $N-2$ points in $\Lambda$.
We use them to define entire functions
\bdm
   f_{ij}(\zeta)
   \,=\, e^{-\beta \thsp\bigl(\!\thsp u(|R_i-R_j|)+\zeta v_0(|R_i-R_j|)\bigr)}-1\,, \qquad
   1\leq i<j\leq N\,,
\edm
of $\zeta\in\C$, and we observe that
\bdm
   \varphi_N(\zeta)
   \,=\, \sum_{\gC\in\gCC_N}\prod_{(i,j)\in\gC}\!\!f_{ij}(\zeta)
\edm
is also an entire function, because the number of connected graphs 
with $N$ vertices is finite. For $0<\eps\leq 1/2$ we can therefore 
apply Cauchy's integral formula to deduce that
\begin{subequations}
\label{eq:varphi-Delta}
\be{varphi-difference}
   \bigl|\varphi_N(\eps)-\varphi_N(0)\bigr|
   \,=\, \Bigl|\frac{\eps}{2\pi\rmi}\int_{|\zeta|=1} 
               \frac{\varphi_N(\zeta)}{\zeta(\zeta-\eps)}\dzeta\Bigr|
   \,\leq\, \frac{\eps}{\pi}
            \int_0^{2\pi} \bigl|\varphi_N(e^{\rmi\theta})\bigr|\dtheta
   \phantom{xx}
\ee
and
\be{varphi-remainder}
\begin{aligned}
   &\bigl|\varphi_N(\eps)-\varphi_N(0)-\eps\varphi_N'(0)\bigr|\\[1ex]
   &\phantom{\bigl|\varphi_N(\eps)-\varphi_N(0)\bigr|}
    \,=\, \Bigl|\frac{\eps^2}{2\pi\rmi}\int_{|\zeta|=1} 
                \frac{\varphi_N(\zeta)}{\zeta^2(\zeta-\eps)}\dzeta\Bigr|
    \,\leq\, \frac{\eps^2}{\pi}
            \int_0^{2\pi} \bigl|\varphi_N(e^{\rmi\theta})\bigr|\dtheta\,.
\end{aligned}
\ee
\end{subequations}
The absolute value of $\varphi_N$ can be 
estimated by means of a tree-graph inequality
\be{tree-graph}
   |\varphi_N(\zeta)|
   \,\leq\, e^{N\beta B}  
            \sum_{\gT\in\gTT_N}\prod_{(i,j)\in\gT}\!
                \bigl|f_{ij}(\zeta)\bigr|\,, \qquad
   |\zeta|\leq 1\,,
\ee
where $\gTT_N$ is the set of trees with $N$ vertices. 
Note that this inequality, which can be found in Ueltschi~\cite{Uelt17}, 
makes use of the stability bound~\req{Ueltschi-stab-bound}, 
which in turn requires $|\zeta|\leq 1$.

By virtue of \req{ftildew} we have the inequality
\be{int-tmp}
   \int_{\Lambda^{N-2}} \prod_{(i,j)\in\gT}\!\bigl|f_{ij}(\zeta)\bigr| \dRR_{2,N}
   \,\leq\, \int_{\R^{N-2}} \prod_{(i,j)\in\gT}\!\! c_\beta w(R_i-R_j) \dRR_{2,N}
\ee
for any fixed tree $\gT\in\gTT_N$. Such a tree consists of (i) a ``backbone'' 
with, say, $n$ bonds and $n-1$ inner vertices, where $1\leq n \leq N-1$, 
which connects the vertices $1$ and $2$, and (ii) $n+1$ subtrees rooted 
at all vertices of this backbone. 
One can first integrate~\req{int-tmp} over all $N-n-1$ vertices 
of these subtrees besides 
their roots, with each of these integrals being bounded by $\cbeta$ according 
to \req{ustar}; 
integrating over the inner vertices of the backbone thereafter
constitutes an $n$-fold autoconvolution of $\cbeta w$, i.e.,
\bdm
   \int_{\Lambda^{N-2}} \!\prod_{(i,j)\in\gT}\!\bigl|f_{ij}(\zeta)\bigr| \dRR_{2,N}
   \,\leq\, \cbeta^{N-1} W_n(R_1-R_2)
   \,=\, \cbeta^{N-1} W_n(R) \,\leq\, \cbeta^{N-1} \Ws(R)\,,
\edm
where we have used \req{monotonicity} for the final inequality.
Note that this estimate is independent of the particular form of the tree $\gT$.
Therefore, making use of Cayley's result that $\gTT_N$ consists of exactly 
$N^{N-2}$ different trees, we deduce from \req{tree-graph} the
upper bound
\be{BehnkeSommer}
   \int_{\Lambda^{N-2}} \bigl|\varphi_N(\zeta)\bigr| \dRR_{2,N}
   \,\leq\, e^{N\beta B} N^{N-2} \cbeta^{N-1} \Ws(R)\,, \qquad |\zeta|\leq 1\,.
\ee

Using the functions $\varphi_N$, $N\geq 2$, we can extend \req{an-graph} 
and \req{an} -- given the fixed coordinate $R\in\Lambda$ --
to scalar functions of a complex variable $\zeta$, namely
\bdm
   \ahat_{N,\Lambda}(\zeta) \,=\, \frac{1}{(N-2)!} 
                            \int_{\Lambda^{N-2}} \varphi_N(\zeta)\thsp\dRR_{2,N}\,,
\edm
and
\be{an-zeta}
   \omegahat_\Lambda(\zeta)
   \,=\, \sum_{N=2}^\infty \ahat_{N,\Lambda}(\zeta)\thsp z^N\,.
\ee
For $\zeta=0$ we recover the original definitions~\req{an-graph} and \req{an}.
Since $\varphi_N$ is absolutely integrable with respect to 
$\RR_{2,N}\in\Lambda^{N-2}$, cf.~\req{BehnkeSommer}, 
and the integral is uniformly bounded for $|\zeta|\leq 1$,
it follows that $\ahat_{N,\Lambda}$ is also complex analytic
for $|\zeta|\leq 1$. 
Furthermore, from \req{varphi-remainder} and \req{BehnkeSommer}
it follows that
\bdm
\begin{aligned}
   \bigl|
      \ahat_{N,\Lambda}(\eps)-\ahat_{N,\Lambda}(0)-\eps \ahat_{N,\Lambda}'(0)
   \bigr|
   &\,\leq\, 2\eps^2\,\frac{N^{N-2}}{(N-2)!} \,e^{N\beta B} \cbeta^{N-1} \Ws(R)
   \\[1ex]
   &\,\leq\, \frac{2\eps^2}{\cbeta}\,(\cbeta e^{\beta B+1})^N\Ws(R)
\end{aligned}
\edm
for $0<\eps\leq 1/2$.
%

Since the infinite series~\req{an-zeta} converges uniformly for
$\zeta\in\C$, $|\zeta|\leq 1$
(for the same fixed parameters $R$, $z$, and the same perturbation 
$v_0\in\V$), the complex extension $\omegahat$ of the Ursell function
is also an analytic function of $\zeta$ in a neighborhood of the unit disk
with
\be{Ursell-remainder-tmp}
   \bigl|
      \omegahat_\Lambda(\eps)-\omegahat_{\Lambda}(0)-\eps\omegahat_{\Lambda}'(0)
   \bigr|
   \,\leq\, \eps^2z^2\,\frac{2\cbeta e^{2(\beta B+1)}}{1-z\cbeta e^{\beta B+1}}\,
            \Ws(R)\,, \qquad
   0 < \eps \leq 1/2\,.
\ee
We already know from \cite{Hank16a} that the original 
Ursell function~\req{Ursell} of $R\in\Lambda$ has a derivative 
$\partial\omega_\Lambda\in \L(\V,L^\infty(\Lambda))$ 
with respect to $u$. Accordingly, \req{Ursell-remainder-tmp} implies that
when choosing $v_0\in\V$, $z$ as in \req{Iz}, and $R\in\Lambda$ as above then
\bdm
   \bigl((\partial \omega_\Lambda)v_0\bigr)(R) \,=\, \omegahat_\Lambda'(0)\,.
\edm
On the other hand, \req{Ursell-remainder-tmp} is valid for every $z$ as in
\req{Iz}, $R\in\Lambda$, and independent of the
particular choice of $v_0\in\V$ with $\norm{v_0}_\V\leq\delta_0$. 
Therefore, denoting by $\omega_\Lambda$ and $\omtilde_\Lambda$ the Ursell 
functions~\req{Ursell} associated with the reference potential $u$ and 
\emph{any} perturbed potential $\utilde=u+v$ with $\norm{v}_\V\leq\delta_0/2$, 
we can rewrite \req{Ursell-remainder-tmp} for
\bdm
   v_0 \,=\, \delta_0 v/\norm{v}_\V \qquad \text{and} \qquad
   \eps\,=\, \norm{v}_\V/\delta_0
\edm
as
\be{Ursell-remainder-Lambda}
   \bigl|
      \bigl(\omtilde_\Lambda - \omega_\Lambda - (\partial\omega_\Lambda)v\bigr) 
      (R)\bigr|
   \,\leq\, z^2\,\frac{2}{\delta_0^2}\,
            \frac{\cbeta e^{2(\beta B+1)}}{1-z\cbeta e^{\beta B+1}}\,
            \norm{v}_\V^2 \Ws(R)\,, 
\ee
valid for every $R\in\Lambda$. 

Starting from \req{an-zeta} with $\zeta=0$ or from \req{varphi-difference},
respectively, the same line of argument leads to the corresponding estimates
\begin{subequations}
\begin{align}
\label{eq:Ursell-bound}
   \bigl|\omega_\Lambda(R)\bigr|
   &\,\leq\, z^2\,\frac{\cbeta e^{2(\beta B+1)}}{1-z\cbeta e^{\beta B+1}}\,\Ws(R)
   \\[1ex]
\intertext{and}
\label{eq:Ursell-Delta}
   \bigl|\omtilde_\Lambda(R) - \omega_\Lambda(R)\bigr|
   &\,\leq\, z^2\,\frac{2}{\delta_0}\,
             \frac{\cbeta e^{2(\beta B+1)}}{1-z\cbeta e^{\beta B+1}}\,
             \norm{v}_\V \Ws(R)\,,
\end{align}
\end{subequations}
valid for every $R\in\Lambda$ and every $v\in\V$ with
$\norm{v}_\V\leq\delta_0/2$.

Before we proceed we note that
from \req{Ursell} and the discussion in Section~\ref{Sec:Grandcanonical} 
one can readily conclude that
\be{Ursell-limit}
   \omega_\Lambda(R) \,\longrightarrow\, \omega(R)
   \,:=\, \rho^{(2)}(R,0)-\rho_0^2\,, \qquad
   |\Lambda|\to\infty\,,
\ee
uniformly on every compact subset of $R\in\R^3$ and
for all activities $z$ in a compact subinterval of \req{Iz}.
We show next that this thermodynamical limit of the Ursell function 
belongs to the Banach space $L_\varrho^\infty(\R^3)$.

\begin{proposition}
\label{Prop:Groeneveld}
Let $u\in\U$ and let $z$ satisfy \req{Iz}, where $\cbeta$ and $B$ are given by
\req{cbeta} and \req{B}, respectively. 
Then the thermodynamical limit of the Ursell function 
belongs to $L_\varrho^\infty(\R^3)$, i.e., there exists $c_\omega>0$ such that
\be{Groeneveld}
   \bigl|\omega(R)\bigr| \,\leq\, c_\omega (1+|R|^2)^{-\alpha/2}
\ee
for every $R\in\R^3$.
\end{proposition}

\begin{proof}
We have mentioned already that
the Mayer function $f$ 
belongs 
to $L_\varrho^\infty(\R^3)$ because of \req{LJtype}
so that the function $w$ defined in \req{ustardef} 
satisfies the assumptions of Proposition~\ref{Prop:Wp-bound} by virtue of 
\req{ustar}. Accordingly, we deduce from Corollary~\ref{Cor:Wp-bound} 
the existence of some constant $C>0$ such that
\be{Ws-konkret}
   \Ws(R) \,\leq\, C (1+|R|^2)^{-\alpha} \qquad \text{for} \quad
   R\in\R^3\,.
\ee
For any fixed activity $z$ from the interval \req{Iz} we therefore obtain 
from \req{Ursell-bound} the inequality
\bdm
   \bigl|\omega_\Lambda(R)\bigr| \,\leq\, c_\omega(1+|R|^2)^{-\alpha}
\edm
for some $c_\omega>0$ and all $R\in\Lambda$. 
Turning to the thermodynamical limit $|\Lambda|\to\infty$
the assertion thus follows from \req{Ursell-limit}.
\end{proof}

Some comments on Proposition~\ref{Prop:Groeneveld} are in order.

The estimate \req{Groeneveld} can be found in a paper by 
Groeneveld~\cite{Groe67} with similar assumptions on the 
pair potential\footnote{The notation in \cite{Groe67} concerning the 
assumptions on $u$ and the corresponding hypothesis
is not fully clear, though.}, 
but it appears that he only published a proof for nonnegative potentials
(in \cite{Groe67b}).
On the other hand, Ruelle included in his book~\cite{Ruel69}
a proof of the weaker statement that $\omega\in L^1(\R^3)$; 
see also \cite{Ruel64}. 

A common way of estimating the decay of the Ursell function consists
in rewriting the Mayer function in \req{tree-graph} as
\bdm
   f(R) \,=\, \bigl(f(R)e^{a(R)}\bigr) e^{-a(R)}
\edm
in such a way that $a$ satisfies a triangle inequality,
and $fe^a$ is bounded and absolutely integrable.
In this case the integral over the backbone considered above can be estimated
by $e^{-a(R)}$ times an autoconvolution of $fe^a$, and this former factor
$e^{-a(R)}$ provides an estimate for the rate of decay. 
In our case this approach could be realized with
\bdm
   e^{a(R)} \,=\, |R|^{\alpha'-3} \qquad 
   \text{for any \ $3<\alpha'<\alpha$}\,,
\edm
but the resulting bound for the Ursell function is evidently suboptimal.

The bound \req{Groeneveld}, on the other hand, is optimal up to 
multiplicative constants, as follows from the cluster expansion~\req{an}, 
which gives
\bdm
   \omega(R) \,=\, z^2 f(R) \,+\, O(z^3)\,, \qquad z\to 0\,,
\edm
according to \req{an-graph-2}.

Now we finalize our investigation of the differentiability of the Ursell
function.

\begin{theorem}
\label{Thm:Ursell-remainder}
Assume that $z$ satisfies \req{Iz}. Then the thermodynamical limit of 
the Ursell function, considered a function of $u\in\U$, has a 
Fr\'echet derivative $\partial\omega\in\L(\V,L^\infty_\varrho(\R^3))$.
More precisely, if $\omega$ and $\omtilde$ denote the thermodynamical limits 
of the Ursell functions corresponding to $u$ and $\utilde=u+v$, respectively,
where $\norm{v}_\V\leq\delta_0(u)/2$,
then there exists $C_\omega=C_\omega(u,z)$, such that
\begin{subequations}
\label{eq:Ursell-limit-Delta}
\begin{align}
\label{eq:Ursell-difference}
   \norm{\omtilde - \omega}_{L^\infty_\varrho(\R^3)}
   &\,\leq\, C_\omega \norm{v}_\V\,,\\[1ex]
\label{eq:Ursell-remainder}
   \norm{\omtilde - \omega - (\partial\omega)v}_{L^\infty_\varrho(\R^3)}
   &\,\leq\, C_\omega \norm{v}_\V^2\,.
\end{align}
\end{subequations}
\end{theorem}

\begin{proof}
Let $u\in\U$ and $v\in\V$ satisfy $\norm{v}_\V\leq\delta_0(u)/2$.
As in \req{Ursell-remainder-Lambda} we write $\omtilde_\Lambda$ for the
Ursell function associated with $\utilde=u+v$ and 
$(\partial\omega_\Lambda)v$ for the derivative of the Ursell function at 
$u$ in direction $v$.
Then \req{Ursell-limit-Delta} readily follows from
\req{Ursell-remainder-Lambda} and \req{Ursell-Delta}
by using \req{Ws-konkret} and turning to
the thermodynamical limit $|\Lambda|\to\infty$, compare \req{Ursell-limit}.
\end{proof}

\section{Iterative Boltzmann inversion}
\label{Sec:IBI}
Now we turn to the fixed point operator 
\bdm
   \Phi(u) \,=\, u \,+\, \gamma \log \frac{F(u)}{\gd}
\edm
associated with the IBI method~\req{IBI}, where $\gamma>0$ is a fixed 
parameter. Recall that $F$ is the nonlinear operator~\req{Henderson-eq}
which takes a pair potential $u\in\U$ onto the corresponding radial
distribution function $g$ in \req{g}. This operator is associated with
the corresponding grand canonical ensemble at fixed (inverse) temperature 
$\beta>0$ and
fixed activity $z>0$, where for technical reasons we slightly restrict
the admissible interval of the activity parameter, cf.~\req{zbound} below.
Note that Henderson~\cite{Hend74} as well as Soper~\cite{Sope96}
considered the operator $F$ for a canonical ensemble at fixed temperature 
and fixed (counting) density $\rho_0$;
it is possible to redo the subsequent analysis also for this case
with $\rho_0$ sufficiently small by using
the equivalence of ensembles and treating
the activity as a function of the counting density and the potential,
i.e., $z=z(\rho_0,u)$.

For fixed radial argument $r>0$ it is an immediate consequence of the results 
in \cite{Hank16a} that the scalar function $u\mapsto(\Phi(u))(r)$ 
is differentiable with respect to $u$, 
and the corresponding derivative is given by
\be{Phiprime}
   \Phi'(u)v \,=\, v \,+\, \gamma\,\frac{F'(u)v}{F(u)}\,,
\ee
pointwise for $r>0$.
However, since $F(u)$ -- as a function of $r$ -- 
decays exponentially near $r=0$ it is not at all obvious
whether $\Phi'$ actually is a Fr\'echet derivative in $\L(\V,\V)$.
This will be established in Theorem~\ref{Thm:Phi} below.
To prepare for this theorem we investigate the core region $0<r\leq r_0$
and the remaining interval $r>r_0$ separately.
We start with the core region.

%
%

\begin{lemma}
\label{Lem:kleine_r}
Let $u\in\U$ and
\be{zbound}
   0 \,<\, z \,\leq\, \zo \,<\, \frac{1}{1+e}\frac{1}{\cbeta e^{2\beta B+1}}
\ee
be arbitrarily fixed. Then there exists $C>0$, depending on $u$ and on $\zo$, 
but independent of $r\in(0,r_0]$, such that
\begin{subequations}
\label{eq:DeltaPhi-klein}
\begin{align}
\label{eq:Delta0Phi-klein}
   \Bigl|\bigl(\Phi(\utilde) - \Phi(u)\bigr)(r)\Bigr|
   &\,\leq\, C\, \norm{\utilde-u}_\V u(r)\,, \\
\label{eq:Delta1Phi-klein}
   \Bigl|\bigl(\Phi(\utilde) - \Phi(u) - \Phi'(u)(\utilde-u)\bigr)(r)\Bigr|
   &\,\leq\, C\, \norm{\utilde-u}_\V^2 \,,
\end{align}
\end{subequations}
uniformly for 
$\utilde\in\U$, provided that $\norm{\utilde-u}_\V$ is sufficiently small.
\end{lemma}

\begin{proof}
Referring to the cavity distribution function $y$ defined in \req{y} we have
\bdm
   F(u) \,=\, g \,=\, e^{-\beta u} y
\edm
from which we deduce the representation 
\be{Fprime}
   F'(u)v \,=\, -\beta e^{-\beta u}vy 
           \,+\, e^{-\beta u}(\partial y)v\,,
\ee
pointwise for $0< r\leq r_0$ and all $v\in\V$;
here, $\partial y$ denotes the Fr\'echet
derivative of $y$ with respect to $u$, compare Proposition~\ref{Prop:rho2-low}.

Let $\sigtilde$ be the cavity distribution function associated with 
$\utilde=u+v$ for some $v\in\V$ sufficiently small. Then we have
\bdmal
   \Phi(\utilde) - \Phi(u)
   &\,=\, v \,+\, \gamma \,\log\frac{F(\utilde)}{F(u)}  
    \,=\, (1-\beta\gamma)v \,+\,
          \gamma\, \log\frac{e^{\beta\utilde}F(\utilde)}{e^{\beta u}F(u)}\\[1ex]
   &\,=\, (1-\beta\gamma)v \,+\,
          \gamma \log(\sigtilde/y)\,,
\edmal
and because $y$ is bounded from below, cf.~Proposition~\ref{Prop:rho2-low}, 
we can use the estimate 
\be{log-bound}
   \bigl|\log(1+x) \,-\, x\bigr| \,\leq\, 2x^2\,, 
   \qquad |x|<1/2\,,
\ee
to obtain
\begin{align}
\label{eq:Phi0diff}
   \Phi(\utilde) - \Phi(u)
   &\,=\, (1-\beta\gamma)v \,+\, \gamma\, \frac{\sigtilde-y}{y}
          \,+\, O(\norm{v}_\V^2)\\[1ex]
\nonumber
   &\,=\, (1-\beta\gamma)v \,+\, O(\norm{v}_\V)\,,
\end{align}
uniformly for $0< r\leq r_0$ and $\norm{v}_\V$ sufficiently small.
Because of \req{LJtype} and the definition~\req{Vd} of $\norm{\,\cdot\,}_\V$ 
this implies assertion~\req{Delta0Phi-klein}.

Starting from \req{Phi0diff} and inserting the representations
\req{Phiprime} and \req{Fprime} of $\Phi'(u)$ and $F'(u)$, respectively, 
it follows that
\bdmal
   &\Phi(\utilde) - \Phi(u) - \Phi'(u)v
    \,=\, \gamma\, \Bigl(\frac{\sigtilde-y}{y} \,-\, \frac{F'(u)v}{F(u)} 
                         \,-\, \beta v
                   \Bigr) 
          \,+\, O(\norm{v}_\V^2)\\[1ex]
   &\qquad \qquad
    \,=\, \frac{\gamma}{y}\, 
          \bigl(\sigtilde - y - e^{\beta u}F'(u)v - \beta vy\bigr)
          \,+\, O(\norm{v}_\V^2)\\[1ex]
   &\qquad \qquad
    \,=\, \frac{\gamma}{y}\, 
          \bigl(\sigtilde - y - (\partial y)v \bigr)
          \,+\, O(\norm{v}_\V^2)
    \,=\, O(\norm{v}_\V^2)
\edmal
by virtue of Proposition~\ref{Prop:rho2-low}, again, and this estimate also
holds uniformly for $0< r \leq r_0$. 
This proves assertion~\req{Delta1Phi-klein}.
\end{proof}

\begin{lemma}
\label{Lem:grosse_r}
Under the assumptions of Lemma~\ref{Lem:kleine_r} there exists $C>0$, such that
\begin{subequations}
\label{eq:DeltaPhi-gross}
\begin{align}
\label{eq:Delta0Phi-gross}
   \Bigl|\bigl(\Phi(\utilde) - \Phi(u)\bigr)(r)\Bigr| 
   &\,\leq\, C \norm{\utilde-u}_\V (1+r^2)^{-\alpha/2}\,,\\
\label{eq:Delta1Phi-gross}
   \Bigl|\bigl(\Phi(\utilde) - \Phi(u) - \Phi'(u)(\utilde-u)\bigr)(r)\Bigr|
   &\,\leq\, C \norm{\utilde-u}_\V^2 (1+r^2)^{-\alpha/2}\,,
\end{align}
\end{subequations}
uniformly for $r\geq r_0$, provided that $\norm{\utilde-u}_\V$ is sufficiently
small.
\end{lemma}

\begin{proof}
Using \req{g} and \req{Ursell-limit} we readily obtain the representation
\be{F2}
   \bigl(F(u)\bigr)(r) \,=\, g(r) \,=\, 1\,+\, \frac{1}{\rho_0^2}\,\omega(R)\,,
   \qquad r=|R|\geq 0\,,
\ee
and therefore \req{Groeneveld}, \req{Ursell-difference}, and the 
differentiability of $\rho_0$ with respect to $u$ imply a local Lipschitz bound
\be{Delta0F}
   \Bigl|\bigl(F(\utilde) - F(u)\bigr)(r)\Bigr|
   \,\leq\, C_F \norm{\utilde-u}_\V (1+r^2)^{-\alpha/2}
\ee
with some $C_F=C_F(u,\zo)>0$ for all $\norm{\utilde-u}_\V$ sufficiently small
and all $r\geq 0$.
Moreover, for $r\geq 0$ and $R\in\R^3$ with $|R|=r$ we further 
deduce from \req{F2} that
\bdm
   \bigl(F'(u)v\bigr)(r)
   \,=\, \partial\Bigl(\frac{1}{\rho_0^2}\,\omega(R)\Bigr)v\,,
\edm
where the right-hand side denotes the derivative
of the scalar function $u\mapsto\omega(R)/\rho_0^2$ 
with respect to $u$ in direction $v\in\V$.
Again, using the differentiability of $\rho_0=\rho_0(u)$, 
Proposition~\ref{Prop:Groeneveld}, and 
Theorem~\ref{Thm:Ursell-remainder}, we conclude that
\be{Delta1F}
   \Bigl|\bigl(F(\utilde) - F(u) - F'(u)v\bigr)(r)\Bigr|
   \,\leq\, C_F' \norm{v}_\V^2 (1+r^2)^{-\alpha/2}\,,
\ee
for some $C_F'>0$, all $v=\utilde-u\in\V$ sufficiently small, and all $r\geq 0$.

Since $F(u)$ is uniformly bounded from below for the given value of the
activity and all $r\geq r_0$ according to Proposition~\ref{Prop:rho2-low}
and \req{LJtype},
\req{Delta0F} implies that the fraction $(F(\utilde)-F(u))/F(u)$ is bounded
by $1/2$ in absolute value, say, for $\norm{v}_\V$ sufficiently small
and all $r\geq r_0$. Accordingly,
\bdm
   \Bigl|\Phi(\utilde) - \Phi(u)\Bigr|
   \,\leq\, |v| \,+\, \gamma \,\Bigl|\log\frac{F(\utilde)}{F(u)}\Bigr|\\[1ex]
   \,=\, |v| \,+\, \gamma \,\frac{|F(\utilde)-F(u)|}{F(u)}
             \,+\, 2\gamma \,\Bigl|\frac{F(\utilde)-F(u)}{F(u)}\Bigr|^2
\edm
by virtue of \req{log-bound}. Using once again that $F(u)$ is bounded from 
below for the respective radii $r\geq r_0$, \req{Delta0F} and \req{Vd} imply 
the first assertion~\req{Delta0Phi-gross}.

Using \req{Phiprime} the same argument as before yields
\bdm
\begin{aligned}
   \bigl|\Phi(\utilde) - \Phi(u) - \Phi'(u)v\bigr|
   &\,=\, \gamma \,\Bigl|\log\frac{F(\utilde)}{F(u)}  
                            \,-\, \frac{F'(u)v}{F(u)}\Bigr|\\[1ex]
   &\,\leq\, \gamma\, \Bigl|\frac{F(\utilde)-F(u)-F'(u)v}{F(u)}\Bigr|
             \,+\, 2\gamma \,\Bigl|\frac{F(\utilde)-F(u)}{F(u)}\Bigr|^2
\edmal
for $\norm{v}_\V$ sufficiently small and all $r\geq r_0$.
The second assertion \req{Delta1Phi-gross} thus follows from \req{Delta1F}
and \req{Delta0F}.
\end{proof}

From Lemma~\ref{Lem:kleine_r} and Lemma~\ref{Lem:grosse_r} 
we immediately conclude our main result.

\begin{theorem}
\label{Thm:Phi}
Let $u\in\U$, and let $z$ satisfy \req{zbound}. 
Then there exists $C_\Phi=C_\Phi(u,\zo)>0$ such that
\bdm
   \norm{\Phi(\utilde)-\Phi(u)}_\V \,\leq\, C_\Phi \norm{\utilde-u}_\V
\edm
for $\norm{\utilde-u}_\V$ sufficiently small. Moreover, 
$\Phi$ is Fr\'echet differentiable with respect to $u$ with 
$\Phi'(u)\in\L(\V,\V)$, and
\bdm
   \norm{\Phi(\utilde)-\Phi(u)-\Phi'(u)(\utilde-u)}_\V
   \,\leq\, C_\Phi \norm{\utilde-u}_\V^2
\edm
for $\norm{\utilde-u}_\V$ sufficiently small.
\end{theorem}

\begin{remark}
\label{Rem:Phi}
\rm
We mention that for the particular choice 
$\gamma=1/\beta$ of the relaxation parameter in IBI, compare~\req{omega1},
the first term on the right-hand side of \req{Phi0diff} cancels, and hence,
in this particular case we have the stronger Lipschitz bounds
\bdmal
   \norm{\Phi(\utilde)-\Phi(u)}_{\Lrho}
   &\,\leq\, C_\Phi \norm{\utilde-u}_\V\,,\\[1ex]
   \norm{\Phi(\utilde)-\Phi(u)-\Phi'(u)(\utilde-u)}_\Lrho
   &\,\leq\, C_\Phi \norm{\utilde-u}_\V^2\,,
\edmal
under the same assumptions as in Theorem~\ref{Thm:Phi}.
\fin
\end{remark}

We now explain how this theorem provides a rigorous justification of IBI.
To this end we assume that the given data $\gd$ in \req{IBI} 
is the true radial distribution function associated with a pair 
potential $\ud\in\U$, i.e.,
\be{gd=Fud}
   \gd \,=\, F(\ud)\,,
\ee
and that $z$ satisfies \req{zbound}.
Further, let $u_0\in\U$ be an initial guess with
$\norm{u_0-\ud}_\Vhat$ sufficiently small.
Then it follows from Theorem~\ref{Thm:Phi} that the first iterate of IBI,
\bdm
   u_1 \,=\, u_0 \,+\, \gamma \log\frac{F(u_0)}{\gd}\,,
\edm
belongs to $\U$ again, because
\bdm
   \norm{u_1-\ud}_\Vhat =\, \norm{\Phi(u_0)-\Phi(\ud)}_\Vhat
   \leq\, C_\Phi\norm{u_0-\ud}_\Vhat\,.
\edm
Accordingly, one can continue iterating and determine further iterates
$u_2,u_3,\dots$. It stays an open problem, however, whether \emph{all} 
iterates will stay within $\U$, or even converge to $\ud$, eventually. 

In practice IBI is usually applied with the potential of mean force
as initial guess. Assuming as before that $\ud$ is a Lennard-Jones type 
pair potential and that $g=\gd$ of \req{gd=Fud} is given exactly, then
this amounts to choosing
\bdm
   u_0 \,=\, -\frac{1}{\beta}\thsp\log \gd\,.
\edm
Since the corresponding cavity distribution function $\sd=e^{\beta\ud}\thbsp\gd$
is bounded and strictly positive according to Proposition~\ref{Prop:rho2-low} 
we have
\begin{subequations}
\label{eq:pmf}
\be{pmf-1}
   -\frac{1}{\beta}\thsp\log \gd(r) 
   \,=\, \ud \,-\, \frac{1}{\beta}\thsp\log\sd \,=\, \ud \,+\, O(1)\,, 
\ee
uniformly for $0<r\leq r_0$. On the other hand it follows from
Proposition~\ref{Prop:Groeneveld} that there exists $c_g>0$ with
\bdm
   \bigl|\gd(r) - 1\bigr| \,\leq\, c_g(1+r^2)^{-\alpha/2}\,, 
\edm
which implies
\be{pmf-2}
   \bigl|\log\gd(r)\bigr| \,\leq\, C(1+r^2)^{-\alpha/2}
\ee
\end{subequations}
by virtue of \req{log-bound}
for some $C>0$ and $r$ sufficiently large. Moreover, since $\sd$ has 
strictly positive lower and upper bounds it follows from the representation 
$\gd=e^{-\beta\ud}\thbsp\sd$ of the radial distribution function
that \req{pmf-2} extends to all $r\geq r_0$ after increasing $C$ appropriately,
when necessary.

We thus conclude from \req{pmf} that the potential
of mean force is a Lennard-Jones type pair potential with the same parameter 
$\alpha$, and therefore the first iteration of IBI is well-defined
for this initial guess.

\section*{Acknowledgements}
The results of this paper have first been presented at the Oberwolfach
Mini-Workshop {\em Cluster Expansions: From Combinatorics to Analysis through 
Probability} (February 2017). 
The author is indebted to Roberto Fern\'{a}ndez, Sabine Jansen, 
and Dimitrios Tsagkarogiannis for the invitation and the
opportunity to contribute this presentation. During the workshop 
David C. Brydges, Aldo Procacci, and Daniel Ueltschi
provided arguments which have considerably simplified our original proof of 
Theorem 5.2. This input and many discussions with further
participants of this workshop are gratefully acknowledged.


\end{document}